\documentclass[a4paper, 10pt]{article}
\usepackage{amsmath,amssymb,amsthm,amsxtra,verbatim,bbm,mathrsfs}

\advance\voffset by -1.5cm
\advance\hoffset by -1.75cm
\newcommand{\Z}{\mathbb{Z}}
\newcommand{\R}{\mathbb{R}}
\newcommand{\C}{\mathbb{C}}
\newcommand{\N}{\mathbb{N}}
\newcommand{\deq}{\mathrel{\mathop:}=}
\newcommand{\eqd}{=\mathrel{\mathop:}}
\renewcommand{\epsilon}{\varepsilon}
\newcommand{\genarg} {\,\cdot\,}

\newcommand{\adj}{^*} 
\newcommand{\mat}[1]{\begin{pmatrix} #1 \end{pmatrix}}

\newcommand{\scalar}[2]{\langle{#1} \mspace{2mu}, {#2}\rangle}

\DeclareMathOperator{\su}{su}

\newcommand{\ket}[1]{| #1 \rangle}

\newcommand{\abs}[1]{\lvert #1 \rvert}
\newcommand{\absb}[1]{\big\lvert #1 \big\rvert}
\newcommand{\absB}[1]{\Big\lvert #1 \Big\rvert}

\newcommand{\norm}[1]{\lVert #1 \rVert}
\newcommand{\normb}[1]{\big\lVert #1 \big\rVert}
\newcommand{\normB}[1]{\Big\lVert #1 \Big\rVert}

\newcommand{\qb}[1]{\bigl[{#1}\bigr]}
\newcommand{\qB}[1]{\Bigl[{#1}\Bigr]}
\newcommand{\qbb}[1]{\biggl[{#1}\biggr]}
\newcommand{\qBB}[1]{\Biggl[{#1}\Biggr]}

\newcommand{\p}[1]{({#1})}
\newcommand{\pb}[1]{\bigl({#1}\bigr)}
\newcommand{\pB}[1]{\Bigl({#1}\Bigr)}

\newcommand{\h}[1]{\{{#1}\}}
\newcommand{\hb}[1]{\bigl\{{#1}\bigr\}}
\newcommand{\hB}[1]{\Bigl\{{#1}\Bigr\}}

\newcommand{\umat}{\mathbbmss{1}} 

\newcommand{\st}{\; : \;}
\newcommand{\qu}[1]{\widehat{#1}}
\newcommand{\quant}{^{\widehat{\mspace{20mu}}}}
\newcommand{\contr}{\rightharpoonup}
\newcommand{\wick}[1]{\,{\mathrel{\mathop:} #1 \mathrel{\mathop:}}\,}

\newtheorem{theorem}{Theorem}
\newtheorem{lemma}{Lemma}

\newtheorem*{remark*}{Remark}
\newtheorem*{remarks*}{Remarks}
\newtheorem*{example}{Example}

\begin{document}

\title{Semi-Classical Dynamics in Quantum Spin Systems}
 
\author{J\"urg Fr\"ohlich$^1$\footnote{juerg@itp.phys.ethz.ch} \qquad Antti Knowles$^1$\footnote{aknowles@itp.phys.ethz.ch} \qquad Enno Lenzmann$^2$\footnote{lenzmann@math.mit.edu} 
\\
\\
\footnotesize \it $^1$Institute of Theoretical Physics, ETH H\"onggerberg,
CH-8093 Z\"urich, 
Switzerland.
\\
\footnotesize \it $^2$Department of Mathematics, Massachusetts Institute of Technology,
Cambridge, MA 02139-4307.
\\
}
\maketitle

\begin{abstract}
We consider two limiting regimes, the large-spin and the mean-field limit, for the dynamical evolution of quantum spin systems. We prove that, in these limits, the time evolution of a class of quantum spin systems is determined by a corresponding Hamiltonian dynamics of classical spins. This result can be viewed as a Egorov-type theorem. We extend our results to the thermodynamic limit of lattice spin systems and continuum domains of infinite size, and we study the time evolution of coherent spin states in these limiting regimes.
\end{abstract}

\section{Introduction}

The purpose of this letter is to study classical limits of quantum spin systems. Work in this direction was undertaken already at the beginning of the seventies. Yet, most of the mathematical results established so far only concern time-independent aspects, such as the classical limit of quantum partition functions for spin systems \cite{Lieb1973, Simon1980}. Here we consider the dynamical evolution of quantum spin systems in limiting regimes; see also \cite{Werner1995}. In particular, we discuss {\em i) the large-spin limit} and {\em ii) the mean-field/continuum limit}. As our main results, we prove that the time evolution of a large class of quantum spin systems approaches the time evolution of classical spins. Our results can be regarded as {\em Egorov-type theorems,} asserting that quantization commutes with time evolution in the classical limit; see \cite{FKP2007} for a similar result on classical and quantum Bose gases. Along the way, we discuss thermodynamic limits and the time evolution of coherent spin states, in the two limits mentioned above.

An example of an evolution equation for classical spins is the \emph{Landau-Lifshitz equation}
\begin{equation} \label{eq:LLintro}
\partial_t M = M \wedge H_{\mathrm{ex}}(M)\,,
\end{equation}
which is widely used in the study of ferromagnetism. Here $M = M(t,x) \in \mathbb{S}^2$ denotes a classical spin field with values on the unit sphere, and $\wedge$ stands for the vector product in $\R^3$. A standard choice for the exchange field is $H_{\mathrm{ex}}(M) = J \, \Delta M$, where $\Delta$ denotes the Laplacian and $J$ is the exchange coupling constant. Equation \eqref{eq:LLintro} then becomes
\begin{equation} \label{eq:LLintro2}
\partial_t M = J \,  M \wedge \Delta M \,.
\end{equation}
This form of the Landau-Lifshitz equation has been studied in the mathematical literature; see for instance \cite{GKT2007, KruzikProhl2006} and references given there. In physical terms, \eqref{eq:LLintro2} describes the dynamics of spin waves in a ferromagnet with nearest neighbor exchange interactions in a classical regime; see \cite{SulemBardos1986}.

In this paper we consider the Landau-Lifshitz equation with an exchange field $H_{\mathrm{ex}}(M)$ given by an integral operator applied to $M$, and generalizations thereof. Equation \eqref{eq:LLintro} then takes the form
\begin{equation} \label{eq:LLintro3}
\partial_t M(t,x) = M(t,x) \wedge \int J(x,y) \, M(t,y) \, dy \,.
\end{equation}
The integral kernel $J(x,y)$ describes the exchange interactions between classical spins beyond the nearest-neighbor approximation in the continuum limit. A formal argument on how to derive \eqref{eq:LLintro2} from \eqref{eq:LLintro3} is given in a remark in Sect.~\ref{sec:dyn:mf}.

Our paper is organized as follows. In Sect.~\ref{section large-spin limit}, we study the dynamics of finite lattice systems of quantum spins in the limit where their spin $s$ approaches $\infty$. The main result of Sect.~\ref{section large-spin limit} is formulated in Theorem \ref{egorov for bounded region} below. In order to prepare the ground for this theorem and its proof, we first introduce a class of Hamilton functions for classical spins and define their quantization by means of a normal-ordering prescription. At the end of Sect.~\ref{section large-spin limit}, we pass to the thermodynamic limit, and we discuss the time evolution of coherent spin states.

In Sect.~\ref{sec:mf:limit}, we present a similar analysis for the mean-field limit of quantum spin systems defined on a lattice with spacing $h > 0$ in the continuum limit, $h \rightarrow 0$. The main result of Sect.~\ref{sec:mf:limit} is stated in Theorem \ref{theorem: continuum egorov for bounded domain} below.

\section{Large-Spin Limit} \label{section large-spin limit}

\subsection{A system of classical spins}
Let $\Lambda$ be a finite subset of the lattice $\Z^d$ (or any other lattice). A classical spin system on $\Lambda$ is described in terms of the finite-dimensional phase space 
\begin{equation*}
\Gamma_\Lambda \;\deq\; \prod_{x \in \Lambda} \mathbb{S}^2\, ,
\end{equation*}
i.\,e.\  we associate an element $M(x)$ of the unit two-sphere $\mathbb{S}^2 \subset \R^3$ with each site $x \in \Lambda$. The phase space $\Gamma_\Lambda$ is conveniently coordinatized as follows. For each site $x \in \Lambda$, let $(M_1(x), M_2(x), M_3(x))$ denote the three cartesian components of a unit vector $M(x) \in \mathbb{S}^2$, and define the complex coordinate functions $(M_+(x), M_z(x), M_-(x))$ on $\mathbb{S}^2$ by 
\begin{equation*}
M_{\pm}(x) \;\deq\; \frac{M_1(x) \pm i M_2(x)}{\sqrt{2}}\,,\qquad M_z(x) \;\deq\; M_3(x) \,.
\end{equation*}
We define a Poisson structure\footnote{Actually $\Gamma_\Lambda$ is symplectic, with symplectic structure determined by the usual one on $\mathbb{S}^2$.} on $\Gamma_\Lambda$ by setting
\begin{equation} \label{poisson bracket in ladder rep}
\{M_i(x), M_j(y)\} \;=\; i \, \tilde{\epsilon}_{i j k} \, \delta(x,y) \, M_k(x)\, .
\end{equation}
Here $\delta(x,y)$ stands for the Kronecker delta, and the indices $i,j,k$ run through the index set $I \deq \{+, z, - \}$, where the symbol $\tilde{\epsilon}_{ijk}$ is defined as $\tilde{\epsilon}_{\pm \mp z} = \pm 1$, $\tilde{\epsilon}_{\pm z \pm} = \mp 1$, $\tilde{\epsilon}_{z \pm \pm} = \pm 1$, and $\tilde{\epsilon}_{ijk} =0 $ otherwise.

For our purposes it is convenient (but not necessary) to replace $\mathbb{S}^2$ with the closed unit ball $\overline{B_1(0)} \subset \R^3$. To this end, we introduce a larger ``phase space'' (a Poisson manifold)
\begin{equation*}
\Xi_\Lambda \;\deq\; \prod_{x \in \Lambda} \overline{B_1(0)}\,,
\end{equation*}
equipped with the $l^\infty$-norm.
The algebra 
\begin{equation*}
\mathfrak{P}_\Lambda \deq \C \big [\{M_i(x) \st i \in I, x \in \Lambda\} \big ]
\end{equation*}
of complex polynomials is a Poisson algebra with Poisson bracket determined by \eqref{poisson bracket in ladder rep}. We equip $\mathfrak{P}_{\Lambda}$ with the norm $\norm{A}_\infty \;\deq\; \sup_{M \in \Xi_\Lambda} \abs{A(M)}$, and we denote its norm closure by $\mathfrak{A}_\Lambda$. Note that, by the Stone-Weierstrass theorem, $\mathfrak{A}_\Lambda$ is the algebra of continuous complex-valued functions on $\Xi_\Lambda$. 

A fairly general class of Hamilton functions $H_\Lambda$ on $\Xi_\Lambda$ may be described as follows: We associate with each multi-index $\alpha \in \N^{I \times \Z^d}$ satisfying $\abs{\alpha} \deq \sum_{x \in \Z^d} \sum_{i \in I} \alpha_i(x) < \infty$
a complex number $V(\alpha)$. Using the trivial embedding $\N^{I \times \Lambda} \subset \N^{I \times \Z^d}$ defined by adjoining zeroes, we consider Hamilton functions of the form
\begin{equation} \label{classical Hamilton function}
H_\Lambda \;\deq\; \sum_{\alpha \in \N^{I \times \Lambda}} V(\alpha)\,M^\alpha\,.
\end{equation}
In order to obtain a real-valued $H_\Lambda$, we require that $\overline{V(\alpha)} = V(\overline{\alpha})$, where the ``conjugation'' $\overline{\alpha}$ of a multi-index $\alpha$ is defined as $\overline{\alpha}_i(x) \deq \alpha_{\overline{i}}(x)$, with $\overline{\genarg}: (+,z,-) \mapsto (-,z,+)$. Furthermore, we impose the following bound on the interaction potential:\footnote{Note that this condition may be weakened by replacing $e^n$ with $e^{rn}$, for any $r > 0$. It may be checked that this does not affect the following results.}
\begin{equation} \label{norm of V}
\norm{V} \;\deq\; \sum_{n \in \N} \sup_{x \in \Z^d} \sum_{\substack{\alpha \in \N^{I \times \Z^d} \\ \abs{\alpha} = n}} \abs{\alpha(x)} \, \abs{V(\alpha)} \, e^n \;<\; \infty\,.
\end{equation}
It is then easy to see that the series \eqref{classical Hamilton function} converges in norm and that the set of allowed interaction potentials $V$ is a Banach space. The Hamiltonian equation of motion reads $\dot{A} = \{H_\Lambda, A\}$, for any observable $A \in \mathfrak{A}_\Lambda$. In particular, a straightforward calculation yields
\begin{equation} \label{hamiltonian equation}
\frac{d}{dt} \,M_i(t,x) \;=\; \sum_{\alpha \in \N^{I \times \Lambda}} V(\alpha) \sum_{j,k} i \tilde{\epsilon}_{jik} \, \alpha_j(x) \, M^{\alpha - \delta_j(x) + \delta_k(x)}(t)\,,
\end{equation}
where the multi-index $\delta_i(x)$ is defined by $[\delta_i(x)]_j(y) \deq \delta_{ij} \delta(x,y)$.

We record the following well-posedness result for the dynamics generated by the class of Hamiltonians introduced above.
\begin{lemma} \label{global well-posedness of ll}
Let $\Lambda$ be a (possibly infinite) subset of $\Z^d$. Let $M_0 \in \Xi_\Lambda$. Then the Hamiltonian equation \eqref{hamiltonian equation} has a unique global-in-time solution $M \in C^1(\R, \Xi_\Lambda)$ that satisfies $M(0) = M_0$. Moreover, the solution $M$ depends continuously on the initial condition $M_0$, and we have the pointwise conservation law $\abs{M(t,x)} = \abs{M(0,x)}$ for all $t$.
\end{lemma}

\begin{proof}
Local-in-time existence and uniqueness follows from a simple contraction mapping argument for the integral equation associated with \eqref{hamiltonian equation}. We omit the details. Also, continuous dependence on $M_0$ follows from standard arguments. Finally, the claim that $\abs{M(0,x)} = \abs{M(t,x)}$ for all $t$ can be easily verified by using \eqref{hamiltonian equation}, which implies that $\frac{d}{dt} M(t,x)$ is perpendicular to $M(t,x)$.
\end{proof}

\begin{remarks*}
{\em 
1. In what follows, we denote the flow map $M_0 \mapsto M(t)$ by $\phi^t_\Lambda$. Note that, under our assumptions, \eqref{hamiltonian equation} also makes sense for infinite $\Lambda \subset \Z^d$, whereas the Hamiltonian $H_\Lambda$ does not have a limit when $\abs{\Lambda} \to \infty$.
\\
2. The last statement implies that the magnitude of each spin remains constant in time, i.\,e.\  the spins precess. In particular, if $M_0 \in \Gamma_\Lambda$, it follows that $M(t) \in \Gamma_\Lambda$ for all $t$. Mathematically, this is simply the statement that the symplectic leaves of the Poisson manifold $\Xi_\Lambda$ remain invariant under the Hamiltonian flow.
\\
3. Time-dependent potentials $V(t,\alpha)$ may be treated without additional complications, provided that the map $t \mapsto V(t)$ is continuous (in the above norm) and $\sup_x$ in \eqref{norm of V} is replaced by $\sup_{x,t}$. The weaker assumption that $t \mapsto V(t,\alpha)$ is continuous for all $\alpha$ implies Lemma \ref{global well-posedness of ll} with the slightly weaker statement that $M \in C(\R, \Xi_\Lambda)$ is a classical solution of \eqref{hamiltonian equation}.
}
\end{remarks*}

\begin{example}
{\em
Consider the Hamiltonian
\begin{equation}
H_\Lambda(t) \;=\; -\sum_{x \in \Lambda} h(t,x) \cdot M(x) 
-\frac{1}{2} \sum_{x,y \in \Lambda} J(x,y) \, M(x) \cdot M(y)\,,
\end{equation}
where $M(x) = (M_1(x), M_2(x), M_3(x))$. Here $h(t,x) \in \R^3$ is an ``external magnetic field'' satisfying $\sup_{t \in \R, x \in \Z^d} \abs{h(t,x)} < \infty$. We also require the map $t \mapsto h(t,x)$ to be continuous for all $x \in \Z^d$. The exchange coupling $J:\Z^d \times \Z^d \rightarrow \R$ is assumed to be symmetric and to satisfy $J(x,x) = 0$ for all $x$. Finally we assume, in accordance with condition \eqref{norm of V}, that $\sup_{x \in \Z^d} \sum_{y \in \Z^d}\abs{J(x,y)}  < \infty$. The corresponding equation of motion for $M(t,x)$ is given by
\begin{equation} \label{landau-lifschitz}
\frac{d}{dt} \, M(t,x) \;=\; M(t,x) \wedge \qbb{h(t,x) + \sum_{y \in \Lambda} J(x,y) M(t,y)}\,,
\end{equation}
i.\,e.\  the \emph{Landau-Lifschitz equation} for a classical lattice spin system. }
\end{example}

\subsection{A system of quantum spins}
In this section we formulate the quantum analogue of the system of classical spins from the previous section. We associate with each point $x \in \Z^d$ a finite-dimensional Hilbert space $\mathscr{H}_x^s \equiv \mathscr{H}_x \deq \C^{2s+1}$ describing a quantum-mechanical spin of magnitude $s$. (Here, and in the following, we refrain from displaying the explicit $s$-dependence whenever it is not needed.) Furthermore, we associate with each finite set $\Lambda \subset \Z^d$ the product space $\mathscr{H}_\Lambda \deq \bigotimes_{x \in \Lambda}\mathscr{H}_x$, and we define the algebra $\qu{\mathfrak{A}}_\Lambda$ as the algebra of (bounded) operators on $\mathscr{H}_\Lambda$, equipped with the operator norm $\| \cdot \|$. 

The spins are represented on $\mathscr{H}_\Lambda$ by a family $\{\qu{S}_i(x) \st i = 1,2,3,\, x \in \Lambda \}$ of operators, where $\qu{S}_i(x)$ is the $i$'th generator of the spin-$s$-representation of $\su(2)$ on $\mathscr{H}_x$, rescaled by $s^{-1}$. In analogy to the complex coordinatization of the classical phase space $\Gamma_\Lambda$ in the previous section, we replace the operators $(\qu{S}_1(x),\qu{S}_2(x),\qu{S}_3(x))$ with $(\qu{S}_+(x),\qu{S}_z(x),\qu{S}_-(x))$ as follows:
\begin{equation*}
\qu{S}_{\pm}(x) \;\deq\; \frac{\qu{S}_1(x) \pm i \qu{S}_2(x)}{\sqrt{2}}\,,\qquad \qu{S}_z(x) \;\deq\; \qu{S}_3 (x)\,, \qquad \text{for all } x \in \Z^d\,.
\end{equation*}
An easy calculation yields  $\norm{\qu{S}_\pm(x)} \leq 1$ and $\norm{\qu{S}_z(x)} = 1$ if $s \geq 1$, and $\norm{\qu{S}_\pm(x)} = \sqrt{2}$ and $\norm{\qu{S}_z(x)} = 1$ if $s = 1/2$.
Furthermore one finds the fundamental commutation relations
\begin{equation} \label{commutation relations in ladder rep}
\qb{\qu{S}_i(x), \qu{S}_j(y)} \;=\; \frac{1}{s} \,\tilde{\epsilon}_{i j k} \, \delta(x,y) \, \qu{S}_k(x)\,
\end{equation}
with $i,j,k \in I$.

\subsection{Quantization}
In order to quantize polynomials in $\mathfrak{P}_\Lambda$ we need a concept of normal ordering.  We say that a monomial $\qu{S}_{i_1}(x_1) \cdots \qu{S}_{i_p}(x_p)$ is \emph{normal-ordered} if $i_k < i_l$ $\Rightarrow$ $k<l$, where $<$ is defined on $I$ through $+< z < -$. We then define normal-ordering by
\begin{equation*}
\wick{\qu{S}_{i_1}(x_1) \cdots \qu{S}_{i_p}(x_p)} \;=\; \qu{S}_{i_{\sigma(1)}}(x_{\sigma(1)}) \cdots \qu{S}_{i_{\sigma(p)}}(x_{\sigma(p)})\,,
\end{equation*}
where $\sigma \in S_p$ is a permutation such that the monomial on the right side is normal-ordered. Next, we define \emph{quantization} $\qu{\genarg}: \mathfrak{P}_\Lambda \rightarrow \qu{\mathfrak{A}}_\Lambda$ by setting
\begin{equation*}
(M_{i_1}(x_1) \cdots M_{i_p}(x_p))\quant \;=\; \wick{\qu{S}_{i_1}(x_1) \cdots \qu{S}_{i_p}(x_p)}\,
\end{equation*}
and by linearity of $\qu{\genarg}$. We set $\qu{1} = \umat$. Note that, by definition, $\qu{\genarg}$ is a linear map (but, of course, not an algebra homomorphism) and satisfies $(\qu{A})^* = \qu{\overline{A}}$. 

\subsection{Dynamics in the large-spin limit}
For each finite $\Lambda \subset \Z^d$ we define the Hamiltonian $\qu{H}_\Lambda$ as the quantization of $H_\Lambda$. More precisely, we quantize \eqref{classical Hamilton function} term by term and note that the resulting series converges in operator norm. Because $H_\Lambda$ is real, the operator $\qu{H}_\Lambda$ is self-adjoint on the finite-dimensional Hilbert space $\mathscr{H}_\Lambda$ and generates a one-parameter group of unitary propagators $U_s(t; \qu{H}_\Lambda)$  (equal to $e^{-is\qu{H}_\Lambda t}$ if $\qu{H}_\Lambda$ is time-independent).

We introduce the short-hand notation
\begin{align*}
\alpha^t_\Lambda A &\;\deq\; A \circ \phi^t_\Lambda \,,\qquad A \in \mathfrak{A}_\Lambda\,,
\\
\qu{\alpha}^t_\Lambda \mathcal{A} &\;\deq\; U_s(t; \qu{H}_\Lambda)^* \, \mathcal{A} \, U_s(t; \qu{H}_\Lambda) \,,\qquad \mathcal{A} \in \qu{\mathfrak{A}}_\Lambda\,,
\end{align*}
where $\phi^t_\Lambda$ is the Hamiltonian flow on $\Xi_\Lambda$.
Note that both $\alpha^t_\Lambda$ and $\qu{\alpha}^t_\Lambda$ are norm-preserving.

We are now able to state and prove our main result for the case of a finite lattice $\Lambda$. Roughly it states that time evolution and quantization commute in the $s \to \infty$ limit. This is a {\em Egorov-type result}.
\begin{theorem} \label{egorov for bounded region}
Let $A \in \mathfrak{P}_\Lambda$ and $\epsilon > 0$. Then there exists a function $A(t) \in \mathfrak{P}_\Lambda$ such that
\begin{equation}
\sup_{t \in \R} \, \norm{\alpha^t_\Lambda A - A(t)}_\infty \;\leq\; \epsilon\,,
\end{equation}
and, for any $t \in \R$,
\begin{equation}
\normb{\qu{\alpha}_\Lambda^t \qu{A}  - \qu{A(t)}} \;\leq\; \epsilon + \frac{C(\epsilon, t, A)}{s}\,,
\end{equation}
where $C(\epsilon,t,A)$ is independent of $\Lambda$.
\end{theorem}

\begin{proof}
Without loss of generality we assume that $A = M^\beta$ for some $\beta \in \N^{I \times \Lambda}$. For simplicity of notation we also assume, here and in the following proofs, that $H_\Lambda$ is time-independent. Consider the Lie-Schwinger series for the time evolution of the classical spin system,
\begin{equation} \label{series for classical spins}
\sum_{l = 0}^\infty \frac{t^l}{l!} \,\hb{H_\Lambda, A}^{(l)}\,,
\end{equation}
where $\hb{H_\Lambda, A}^{(l)} = \hB{H_\Lambda, \hb{H_\Lambda, A}^{(l - 1)}}$ and $\hb{H_\Lambda, A}^{(0)} = A$.
In order to compute the nested Poisson brackets we observe that
\begin{equation}
\h{M^\alpha, M^\beta} \;=\; \sum_{x \in \Lambda} \sum_{i,j,k \in I} i \tilde{\epsilon}_{ijk} \, \alpha_i(x) \, \beta_j(x) \, M^{\alpha + \beta - \delta_i(x) - \delta_j(x) + \delta_k(x)}\,,
\end{equation}
as can be seen after a short calculation. Iterating this identity yields
\begin{align}
&\hb{H_\Lambda, A}^{(l)} \;=\; i^l \sum_{\alpha^1, \dots, \alpha^l} \sum_{x_1,\dots,x_l} \sum_{i_1,\dots, i_l}\sum_{j_1,\dots,j_l}\sum_{k_1,\dots,k_l}
\notag \\
&\qquad \qBB{\prod_{q = 1}^l \tilde{\epsilon}_{i_q j_q k_q} \, V(\alpha^q) \, \alpha^q_{i_q}(x_q) \, \qbb{\beta + \sum_{r = 1}^{q - 1} \pb{\alpha^r - \delta_{i_r}(x_r) - \delta_{j_r}(x_r) + \delta_{k_r}(x_r)}}_{j_q} \!\!\!(x_q)}
\notag \\ \label{multiple Poisson bracket}
&\qquad M^{\beta + \sum_{r = 1}^l \pb{\alpha^r - \delta_{i_r}(x_r) - \delta_{j_r}(x_r) + \delta_{k_r}(x_r)}}\,.
\end{align}
In order to estimate this series, we recall that $\norm{M^\gamma}_\infty \leq 1$ and rewrite it by using that
\begin{equation*}
\sum_{\alpha^1, \dots, \alpha^l} = \sum_{n_1, \dots n_l = 1}^\infty \sum_{\abs{\alpha^1} = n_1}\dots \sum_{\abs{\alpha^l} = n_l}\,.
\end{equation*}
We then proceed recursively, starting with the sum over $\alpha^l,x_l, i_l, j_l,k_l$ and, at each step, using that
\begin{equation*}
\sum_{\abs{\alpha} = n} \sum_x \sum_{i,j,k} \, \abs{\tilde{\epsilon}_{ijk}} \, \alpha_i(x) \, \gamma_j(x) \, \abs{V(\alpha)} \;\leq\; \abs{\gamma} \, \norm{V}^{(n)}\,,
\end{equation*}
where 
\begin{equation*}
\norm{V}^{(n)} \;\deq\; \sup_{x \in \Z^d} \sum_{\abs{\alpha} = n} \abs{V(\alpha)} \, \abs{\alpha(x)}\,.
\end{equation*}
In this manner we find that
\begin{align*}
&\normB{\hb{H_\Lambda, A}^{(l)}}_\infty 
\\
&\qquad \leq\; \sum_{n_1, \dots, n_l = 1}^\infty \, \abs{\beta} (\abs{\beta} + n_1) \cdots (\abs{\beta} + n_1 + \cdots + n_{l - 1}) \, \norm{V}^{(n_1)} \cdots \norm{V}^{(n_l)}
\\
&\qquad \leq\; \sum_{n_1, \dots, n_l = 1}^\infty \, (\abs{\beta} + n_1 + \cdots + n_l)^l \, \norm{V}^{(n_1)} \cdots \norm{V}^{(n_l)}
\\
&\qquad \leq\; l! \sum_{n_1, \dots, n_l = 1}^\infty \, e^{\abs{\beta} + n_1 + \cdots + n_l} \, \norm{V}^{(n_1)} \cdots \norm{V}^{(n_l)}
\\
&\qquad =\; l! \, e^{\abs{\beta}} \, \norm{V}^l\,.
\end{align*}
Thus, for $\abs{t} < \norm{V}^{-1}$, the series \eqref{series for classical spins} converges in norm, and an analogous estimate of the remainder of the Lie-Schwinger expansion of $\alpha^t_\Lambda A$ shows that \eqref{series for classical spins} equals $\alpha^t_\Lambda A$. As all estimates are independent of $\Lambda$, the convergence is uniform in $\Lambda$.

The quantum-mechanical case is similar. Consider the Lie-Schwinger series for the time evolution of the quantum spin system:
\begin{equation} \label{series for qm spins}
\sum_{l = 0}^\infty \frac{t^l}{l!}\, (is)^l \, \qb{\qu{H}_\Lambda, \qu{A}}^{(l)}\,,
\end{equation}
where $\qb{\qu{H}_\Lambda, \qu{A}}^{(l)} = \qB{\qu{H}_\Lambda, \qb{\qu{H}_\Lambda, \qu{A}}^{(l - 1)}}$  and $\qb{\qu{H}_\Lambda, \qu{A}}^{(0)} = \qu{A}$. 
In order to estimate the multiple commutators, we remark that, from \eqref{poisson bracket in ladder rep} and \eqref{commutation relations in ladder rep} and since both $\{\genarg,\genarg\}$ and $is [\genarg,\genarg]$ are derivations in both arguments, we see that $(is)^l\qb{\qu{H}_\Lambda, \qu{A}}^{(l)}$ is equal to the expression obtained from $\hb{H_\Lambda, A}^{(l)}$ by reordering the terms appropriately and by replacing $M_i(x)$ with $\qu{S}_i(x)$. In particular (assuming $s \geq 1$)
\begin{equation*}
\normB{(is)^l\qb{\qu{H}_\Lambda, \qu{A}}^{(l)}} \;\leq\; l! \, e^{\beta} \, \norm{V}^l\,,
\end{equation*}
and we deduce exactly as above that \eqref{series for qm spins} equals $\qu{\alpha}_\Lambda^t \qu{A}$ for $t < \norm{V}^{-1}$.

To show the claim of the theorem for $\abs{t} < \norm{V}^{-1}$ we first remark that $\qu{\alpha_\Lambda^t(A)}$ is well-defined through its convergent power series expansion. Now as shown above, each term of $(is)^l\qb{\qu{H}_\Lambda, \qu{A}}^{(l)}$, as a polynomial in $\qu{\mathfrak{A}}_\Lambda$, is equal to a reordering of the corresponding term of $\hb{H_\Lambda, A}^{(l)}$. If $P$ is a monomial (with coefficient 1) of degree $p$ in the generating variables $\{\qu{S}_i(x)\}$ and $\tilde{P}$ a monomial obtained from $P$ by any reordering of terms, the commutation relations \eqref{commutation relations in ladder rep} imply that
\begin{equation*}
\norm{P - \tilde{P}} \;\leq\; \frac{p^2}{s}\,.
\end{equation*}
Thus
\begin{equation*}
\pB{\hb{H_\Lambda, A}^{(l)}}\quant \;=\;
(is)^l\qb{\qu{H}_\Lambda, \qu{A}}^{(l)} + R_l\,,
\end{equation*}
where, recalling the expression \eqref{multiple Poisson bracket} and the estimates following it, we see that the ``loop terms'' $R_l$ are bounded by
\begin{align*}
\norm{R_l} &\;\leq\; \frac{1}{s}\sum_{n_1, \dots, n_l = 1}^\infty \, (\abs{\beta} + n_1 + \cdots + n_l)^{l+2} \, \norm{V}^{(n_1)} \cdots \norm{V}^{(n_l)}
\\
&\;\leq\; \frac{(l+2)!}{s} \, e^{\abs{\beta}} \, \norm{V}^l\,.
\end{align*}
Therefore, if $\abs{t} < \norm{V}^{-1}$,
\begin{equation*}
\norm{\qu{\alpha}_\Lambda^t \qu{A}  - \qu{\alpha_\Lambda^t A }} \;\leq\; \frac{e^{\abs{\beta}}}{s} \sum_{l = 0}^\infty (l+2) (l+1) (t \norm{V})^l \;\leq\; \frac{C(t,A)}{s}\,,
\end{equation*}
where $C(t,A)$ is independent of $\Lambda$.

In order to extend the result to arbitrary times we proceed by iteration. The crucial observations that enable this process are that the convergence radius $\norm{V}^{-1}$ is independent of $\abs{\beta}$ and $\alpha^t_\Lambda, \qu{\alpha}^t_\Lambda$ are norm-preserving. Let $t \in \R$ and choose $\nu \in \N$ such that $\tau \deq t / \nu$ satisfies $\abs{\tau} < \norm{V}^{-1}$.
In order to iterate we need to introduce a cutoff in the series \eqref{series for classical spins} and \eqref{multiple Poisson bracket}. The series \eqref{series for classical spins} consists of an infinite sum of terms in $\mathfrak{P}_\Lambda$ which are be indexed by $(l, \alpha^1, \dots, \alpha^l,x_1, \dots, x_l)$. Now let $\epsilon > 0$ be given. Since the series converges in norm there is a finite subset 
\begin{equation*}
B_1 \;=\; B_1(\epsilon) \;\subset\; \{(l, \alpha^1, \dots, \alpha^l,x_1, \dots, x_l)\} \;=\; \bigcup_{l = 0}^\infty \pb{\N^{I \times \Lambda}}^l \times \Lambda^l
\end{equation*}
such that the norm of the series restricted to the complement of $B_1$ is smaller than $\epsilon/\nu$. This induces a splitting
$
\alpha_\Lambda^\tau A = \alpha_{B_1} A + \alpha_{B_1^c} A
$
(in self-explanatory notation), such that $\alpha_{B_1} A \in \mathfrak{P}_\Lambda$ and $\norm{\alpha_{B_1^c} A}_\infty \leq \epsilon/\nu$.
Similarly, one splits
$
\qu{\alpha}_\Lambda^\tau \qu{A} = \qu{\alpha}_{B_1} \qu{A} + \qu{\alpha}_{B_1^c} \qu{A}
$
where, after an eventual increase of $B_1$, $\norm{\qu{\alpha}_{B_1^c} \qu{A}} \leq \epsilon/\nu$.

Now we use the above result for $\abs{\tau} < \norm{V}^{-1}$:
\begin{equation*}
\qu{\alpha}_\Lambda^\tau \qu{A} \;=\; \qu{\alpha_{B_1}A} + \frac{R_1}{s} + \qu{\alpha}_{B_1^c} \qu{A}
\end{equation*}
where $R_1$ is some bounded operator. Since $\alpha_{B_1} A \in \mathfrak{P}_\Lambda$ we may repeat the process on the time interval $[\tau, 2\tau]$:
\begin{align*}
\qu{\alpha}_\Lambda^\tau \, \qu{\alpha}_\Lambda^\tau \qu{A} &\;=\; \qu{\alpha}_\Lambda^\tau \,\qu{\alpha_{B_1}A} + \frac{\qu{\alpha}_\Lambda^\tau \,R_1}{s} + \qu{\alpha}_\Lambda^\tau \,\qu{\alpha}_{B_1^c} \qu{A}
\notag \\
&\;=\; \pb{\alpha_{B_2} \,\alpha_{B_1}A}\quant + \qu{\alpha}_{B_2^c} \,\qu{\alpha_{B_1}A} + \qu{\alpha}_\Lambda^\tau \,\qu{\alpha}_{B_1^c} \qu{A} + \frac{R_2 + \qu{\alpha}_\Lambda^\tau \,R_1}{s}
\end{align*}
Continuing in this manner one sees that, since $\alpha_\Lambda^t$ and $\qu{\alpha}_\Lambda^t$ are norm-preserving, $A(t) \;\deq\; \alpha_{B_\nu} \cdots \alpha_{B_1} A \;\in\; \mathfrak{P}_\Lambda$ satisfies
\begin{equation*}
\norm{\alpha_\Lambda^t A - A(t)}_\infty \;\leq\; \epsilon\,,
\end{equation*}
as well as
\begin{equation*}
\norm{\qu{\alpha}_\Lambda^t \qu{A} - \qu{A(t)}} \;\leq\; \epsilon + \frac{C(\epsilon, t, A)}{s}\,.
\end{equation*}
\end{proof}

\subsection{The thermodynamic limit} \label{section: thermodynamic limit for lattice}
The above analysis was done for a finite subset $\Lambda$, but the observed uniformity in $\Lambda$ allows for a statement of the result directly in limit $\Lambda = \Z^d$. We pause to describe how this works.

Concentrate first on the quantum case. If $\Lambda_1 \subset \Lambda_2$, an operator $\mathcal{A}_1 \in \qu{\mathfrak{A}}_{\Lambda_1}$ may be identified in the usual fashion with an operator $\mathcal{A}_2 \in \qu{\mathfrak{A}}_{\Lambda_2}$ by setting $\mathcal{A}_2 = \mathcal{A}_1 \otimes \umat_{\Lambda_2 \setminus \Lambda_1}$. We shall tacitly make use of this identification in the following. It induces the norm-preserving mapping $\qu{\mathfrak{A}}_{\Lambda_1} \rightarrow \qu{\mathfrak{A}}_{\Lambda_2}$ of the abstract $C^*$-algebras and the isotony relation $\qu{\mathfrak{A}}_{\Lambda_1} \subset \qu{\mathfrak{A}}_{\Lambda_2}$. Observables of the quantum spin system in the thermodynamic limit are then elements of the \emph{quasi-local algebra}
\begin{equation*}
\qu{\mathfrak{A}} \;\deq\; \overline{\bigvee_{\Lambda \subset \Z^d \text{ finite}} \qu{\mathfrak{A}}_\Lambda}\,,
\end{equation*}
which is the $C^*$-algebra defined as the closure of the normed algebra generated by the union of all $\qu{\mathfrak{A}}_\Lambda$'s, where $\Lambda$ is finite. The spins are represented on $\qu{\mathfrak{A}}$ by a family $\{\qu{S}_i(x) \st i \in I,\, x\in \Z^\nu\}$ of operators.

The dynamics of the system is determined by a one-parameter group $(\qu{\alpha}^t)_{t \in \R}$ of automorphisms of $\qu{\mathfrak{A}}$. Its existence is a corollary of the proof of Theorem \ref{egorov for bounded region}.
\begin{lemma} \label{quantum limit dynamics}
Let $\mathcal{A} \in \qu{\mathfrak{A}}_{\Lambda_0}$ for some finite $\Lambda_0 \subset \Z^d$ and $t \in \R$. Then the following limit exists in the norm sense:
\begin{equation*}
\lim_{\Lambda \to \infty} \qu{\alpha}_\Lambda^t \mathcal{A} \;\eqd\; \qu{\alpha}^t \mathcal{A} \,,
\end{equation*}
where $\Lambda \to \infty$ means that $\Lambda$ eventually contains every finite subset.
By continuity this extends to a strongly continuous one-parameter group $(\qu{\alpha}^t)_{t \in \R}$ of automorphisms of $\qu{\mathfrak{A}}$.
\end{lemma}
\begin{proof}
For $\abs{t} < \norm{V}^{-1}$ the series \eqref{series for qm spins} is bounded in norm, uniformly in $\Lambda$, so to show convergence of the series it suffices to show the convergence of $\qb{\qu{H}_\Lambda, \qu{A}}^{(l)}$ for each $l \in \N$, which is an easy exercise.

Thus $\qu{\alpha}^t \mathcal{A}$ is well-defined for any polynomial $\mathcal{A}$. By continuity, $\qu{\alpha}^t$ extends to an automorphism of $\qu{\mathfrak{A}}$. Since $\qu{\alpha}^t \mathcal{A} \in \qu{\mathfrak{A}}$ and $\qu{\alpha}^t$ is a one-parameter group, we may extend it to all times by iteration. Strong continuity follows since $\qu{\alpha}^t \mathcal{A}$, for small $t$ and polynomial $\mathcal{A}$, is defined through a convergent power series:
\begin{equation*}
\lim_{t \to 0} \, \norm{\qu{\alpha}^t \mathcal{A} - \mathcal{A}} \;=\; 0\,.
\end{equation*}
By continuity, this remains true for all $\mathcal{A} \in \qu{\mathfrak{A}}$.
\end{proof}

For classical spin systems we recall that, for finite $\Lambda$, we have $\mathfrak{A}_\Lambda = C\pb{\prod_{x \in \Lambda} \overline{B_1(0)};\C}$,
a $C^*$-algebra under $\norm{\genarg}_\infty$. As above, for  $\Lambda_1 \subset \Lambda_2$, we identify $A_1 \in \mathfrak{A}_{\Lambda_1}$  with a function $A_2 \in \mathfrak{A}_{\Lambda_2}$ by setting $A_2 = A_1 \otimes 1_{\Lambda_2 \setminus \Lambda_1}$. We thus get a norm-preserving mapping $\mathfrak{A}_{\Lambda_1} \rightarrow \mathfrak{A}_{\Lambda_2}$ of the abstract $C^*$-algebras and the relation $\mathfrak{A}_{\Lambda_1} \subset \mathfrak{A}_{\Lambda_2}$. Define the classical quasi-local algebra as
\begin{equation*}
\mathfrak{A} \;\deq\; \overline{\bigvee_{\Lambda \subset \Z^d \text{ finite}} \mathfrak{A}_\Lambda}\,.
\end{equation*}
Note that $\mathfrak{A}$ is equal to the space of continuous complex functions on $\prod_{x \in \Z^d} \overline{B_1(0)}$, equipped with the product topology (this is an immediate consequence of the Tychonoff and Stone-Weierstrass theorems).

The spins are represented on $\mathfrak{A}$ by a family $\{M_i(x) \st i \in I,\, x\in \Z^\nu\}$ of functions. Existence of the dynamics follows exactly as above.
\begin{lemma}
Let $A \in \mathfrak{A}_{\Lambda_0}$ for some finite $\Lambda_0 \subset \Z^d$ and $t \in \R$. Then the following limit exists in $\norm{\genarg}_\infty$:
\begin{equation*}
\lim_{\Lambda \to \infty} \alpha_\Lambda^t A \;\eqd\; \alpha^t A \,.
\end{equation*}
By continuity this extends to a strongly continuous one-parameter group $(\alpha^t)_{t \in \R}$ of automorphisms of $\mathfrak{A}$.
Furthermore, $\alpha^t A = A \circ \phi^t$, where $\phi^t = \phi^t_{\Z^d}$ is the Landau-Lifschitz flow defined in the remark after Lemma \ref{global well-posedness of ll}.
\end{lemma}

Now set $\mathfrak{P} \deq \C[\{M_i(x)\st i \in I,\, x \in \Z^d\}]$.
Then the proof of Theorem \ref{egorov for bounded region} yields the following
\begin{theorem} \label{egorov for thermodynamic limit}
Let $A \in \mathfrak{P}$ and $\epsilon > 0$. Then there exists a function $A(t) \in \mathfrak{P}$ such that
\begin{equation}
\sup_{t \in \R} \, \norm{\alpha^t A - A(t)}_\infty \;\leq\; \epsilon\,,
\end{equation}
and, for any $t \in \R$,
\begin{equation}
\normb{\qu{\alpha}^t \qu{A}  - \qu{A(t)}} \;\leq\; \epsilon + \frac{C(\epsilon, t, A)}{s}\,.
\end{equation}
\end{theorem}
\begin{remark*}
{\em 
In particular, the result applies to classical equations of motion of the form \eqref{landau-lifschitz} where the sum over $y$ ranges over $\Z^d$.
}
\end{remark*}

\subsection{Evolution of coherent states}
Denote by $S_i \deq s \qu{S}_i$ the unscaled spin operator in the spin-$s$-representation of $\su(2)$.
For the polar angles $(\theta, \varphi) \in [0,\pi] \times [0,2\pi)$ corresponding to the unit vector $M \in \mathbb{S}^2$ we define the coherent state in $\C^{2s +1}$ as
\begin{equation*}
\ket{M} \;\deq\; \exp \frac{\theta}{\sqrt{2}} \qB{e^{i\varphi}\, S_- - e^{-i\varphi} \, S_+} \ket{s}\,,
\end{equation*}
where $\ket{s}$ is the highest-weight state, i.\,e.\  $S_z \,\ket{s} = s \,\ket{s}$. Note that $\ket{M} = e^{i\alpha \cdot S} \, \ket{s}$, where $\alpha = \theta \, n$ and $n$ is the unit vector $(\sin \varphi, - \cos \varphi, 0)$.

Set
$
A \deq \frac{\theta}{\sqrt{2}} \qB{e^{i\varphi}\, S_- - e^{-i\varphi} \, S_+}
$
and $U \deq e^A$ so that $\ket{M} = U \ket{s}$.
Then using
\begin{equation*}
U\adj \,S_i \, U = \sum_{k = 0}^\infty \frac{1}{k!}\, [\dots[S_i, A],\dots,A]\, ,
\end{equation*}
we find
\begin{align}
U\adj \, S_1 \, U &\;=\; \sin \theta \cos \varphi \,S_z + \frac{1}{\sqrt{2}}\cos^2 \frac{\theta}{2} \,(S_+ + S_-) - \frac{1}{\sqrt{2}} \sin^2 \frac{\theta}{2} \, \pb{e^{-2i\varphi}S_+ + e^{2i\varphi}S_-}\,,
\notag \\
U\adj \, S_2 \, U &\;=\; \sin \theta \sin \varphi \,S_z + \frac{1}{\sqrt{2}i}\cos^2 \frac{\theta}{2} \,(S_+ - S_-) - \frac{1}{\sqrt{2}i} \sin^2 \frac{\theta}{2} \, \pb{e^{2i\varphi}S_+ - e^{-2i\varphi}S_-}\,,
\notag \\ \label{rotated spins}
U\adj \, S_3 \, U &\;=\;\cos \theta \, S_z - \frac{1}{\sqrt{2}} \sin \theta \, \pb{e^{-i \varphi} S_+ + e^{i \varphi} S_-}\,.
\end{align}
As a consequence note that
\begin{equation}\label{sandwich of coherent states}
\scalar{M} {\qu{S} \, M} \;=\; \mat{\sin \theta \cos \varphi \\ \sin \theta \sin \varphi \\ \cos \theta} \;=\; M\,.
\end{equation}

In order to derive our main result for coherent spins states we need the following lemma, which follows from direct calculations.
\begin{lemma} For any unit vector $M \in \mathbb{S}^2$, we have that \label{properties of coherent spin states}
\begin{equation*} 
\absB{\scalar{M}{\qu{S}_{i_1} \cdots \qu{S}_{i_p} \, M} - M_{i_1} \cdots M_{i_p}} \;\leq\; p \, \sqrt{\frac{2}{s}}\,.
\end{equation*}
\end{lemma}

Now let $M : \Z^d \rightarrow \mathbb{S}^2$ be a configuration of classical spins on the lattice. Then $M$ defines a state $\rho_M$ on $\qu{\mathfrak{A}}$ as follows. For finite $\Lambda$, consider the product state
\begin{equation*}
\ket{M_\Lambda} \;\deq\; \bigotimes_{x \in \Lambda} \ket{M(x)} \in \mathscr{H}_\Lambda\,.
\end{equation*}
Then, for $\mathcal{A} \in \mathfrak{\qu{A}}_\Lambda$, we set
\begin{equation*}
\rho_M (\mathcal{A}) \deq \scalar{M_\Lambda}{\mathcal{A} \, M_\Lambda }\,,
\end{equation*}
and extend the definition of $\rho_M$ to arbitrary $\mathcal{A} \in \mathfrak{\qu{A}}$ by continuity.

Let $M: \R \times \Z^d \rightarrow \mathbb{S}^2$ be the solution of the Hamiltonian equation of motion \eqref{hamiltonian equation} with initial conditions $M(0,x) = M(x)$. The following result links the quantum time evolution for coherent spin states with the corresponding classical configuration in the large-spin limit.

\begin{theorem}
Let $t \in \R$, $A \in \mathfrak{P}$ and $M : \Z^d \rightarrow \mathbb{S}^2$. Then
\begin{equation*}
\lim_{s \to \infty} \rho_M(\qu{\alpha}^t \qu{A})
\;=\; A(M(t))\,,
\end{equation*}
uniformly in $t$ on compact time intervals.
\end{theorem}
\begin{proof}
The proof is essentially a corollary of Lemma \ref{properties of coherent spin states} and the proof of Theorem \ref{egorov for bounded region}. First, let $\abs{t} < \norm{V}^{-1}$. We know from \eqref{series for classical spins} and \eqref{series for qm spins} that
\begin{equation*}
A(M(t)) \;=\; \sum_{l = 0}^\infty \frac{t^l}{l!} \, \lim_{\Lambda \to \infty} \, \hb{H_\Lambda, A}^{(l)} (M)\,,
\end{equation*}
as well as
\begin{equation*}
\rho_M(\qu{\alpha}_\Lambda^t \qu{A}) \;=\; \sum_{l = 0}^\infty \frac{t^l}{l!}\, (is)^l \, \rho_M\pB{\lim_{\Lambda \to \infty}\,\qb{\qu{H}_\Lambda, \qu{A}}^{(l)}}\,.
\end{equation*}
Now Lemma \ref{properties of coherent spin states} implies that
\begin{equation*}
\absb{\rho_M\pb{\qu{M^\alpha}} - M^\alpha} \leq \abs{\alpha} \, \sqrt{\frac{2}{s}}\,.
\end{equation*}
Arguing exactly as in the proof of Theorem \ref{egorov for bounded region}, we get the bound
\begin{equation}
\absb{\rho_M(\qu{\alpha}^t \qu{A}) - A(M(t))} \;\leq\; \frac{C(t, A)}{\sqrt{s}}\,.
\end{equation}
Arbitrary times are reached by iteration as in the proof of Theorem \ref{egorov for bounded region}. \end{proof}

\section{Mean-Field Limit} \label{sec:mf:limit}

This section is devoted to the dynamics of a quantum spin system in the mean-field/continuum limit. More precisely, we consider a system of quantum spins on a lattice with spacing $h >0$. The limit $h \rightarrow 0$ yields again a Egorov-type result: The quantum dynamics approaches the dynamics of a classical spin system defined on a continuum set. As in the previous section, we also discuss the thermodynamic limit and the time evolution of coherent states.

\subsection{A system of quantum spins on a lattice}
Let $\Lambda \subset \R^d$ be bounded and open. We associate with each spacing $h >0$ the finite lattice 
\begin{equation*}
\Lambda^{(h)} \deq h \Z^d \cap \Lambda. 
\end{equation*}
At each lattice site $x \in \Lambda^{(h)}$ there is a spin of (fixed) magnitude $s$. The Hilbert space of this quantum system is
\begin{equation*}
\mathscr{H}^{(h)}_\Lambda \;\deq\; \bigotimes_{x \in \Lambda^{(h)}} \C^{2s+1}\,.
\end{equation*}
The algebra of bounded operators on $\mathscr{H}_\Lambda^{(h)}$ is denoted by $\qu{\mathfrak{A}}_\Lambda^{(h)}$.

The spins are represented on $\mathscr{H}^{(h)}_\Lambda$ by a family $\{\qu{S}_i(x) \st i = 1,2,3,\, x \in \Lambda^{(h)} \}$ of operators, where $\qu{S}_i(x)$ is the $i$'th generator of the spin-$s$-representation of $\su(2)$, rescaled by $h^d/s$. As usual, we replace the operators $(\qu{S}_1,\qu{S}_2,\qu{S}_3)$ with $(\qu{S}_+,\qu{S}_z,\qu{S}_-)$. They satisfy the bounds
$\norm{\qu{S}_\pm} \leq h^d$ and $\norm{\qu{S}_z} \;=\; h^d$ if $s \geq 1$, as well as
$\norm{\qu{S}_\pm} = \sqrt{2} h^d$ and $\norm{\qu{S}_z} \;=\; h^d$ if $s = 1/2$.
The commutation relations now read
\begin{equation} \label{continuum commutation relations in ladder rep}
\qb{\qu{S}_i(x), \qu{S}_j(y)} \;=\; \frac{h^d}{s} \,\tilde{\epsilon}_{i j k} \, \delta(x,y) \, \qu{S}_k(x)\,,
\end{equation}
with $i,j,k \in I$.

\subsection{A continuum theory of spins}
Let $\Lambda \subset \R^d$ be a bounded, open set. A system of classical spins on $\Lambda$ is represented in terms of the Poisson ``phase space''\footnote{As in the previous section, one may introduce a symplectic phase space $\Gamma_\Lambda$ consisting of all $M \in \Xi_\Lambda$ such that $\abs{M(x)} = 1$ a.e.}
\begin{equation*}
\Xi_\Lambda \;\deq\; \hb{M \in L^\infty(\Lambda; \R^3)\st \norm{M}_\infty \leq 1}\,,
\end{equation*}
which we equip with the $L^\infty$-norm.
In analogy to Section \ref{section large-spin limit}, we use the complex coordinates $(M_+,M_z,M_-)$ instead of $(M_1, M_2,M_3)$, so that the Poisson bracket on $\Xi_\Lambda$ satisfies
\begin{equation} \label{ladder continuum poisson bracket}
\hb{M_i(x), M_j(y)} \;=\; i \, \tilde{\epsilon}_{ijk} \, \delta(x - y) \, M_k(x)\,,
\end{equation}
for $i,j,k \in I$.

In order to describe a useful class of observables on $\Xi_\Lambda$, we introduce the space $\mathscr{B}^{(p)}$, $p \in \N$, which consists of all functions $f$ in $C(\R^{pd};\C^{3^p})$ that are symmetric in their arguments, in the sense that $Pf = f$, where
\begin{equation*}
(Pf)_{i_1 \dots i_p}(x_1, \dots, x_p) \;\deq\; \frac{1}{p!} \sum_{\sigma \in S_p} f_{i_{\sigma(1)} \dots i_{\sigma(p)}}(x_{\sigma(1)}, \dots, x_{\sigma(p)})\,.
\end{equation*}
On the space $\mathscr{B}^{(p)}$ we introduce the norms
\begin{equation*}
\norm{f}_1^{(h)} \;\deq\; h^{pd} \sum_{i_1, \dots, i_p \in I} \sum_{x_1,\dots, x_p \in h \Z^d} \abs{f_{i_1\dots i_p}(x_1,\dots, x_p)}\,,
\end{equation*}
\begin{equation*}
\norm{f}_{\infty, 1}^{(h)} \;\deq\; \sup_{x} \sum_{i_1, \dots, i_p \in I} \, h^{(p - 1)d} \sum_{x_2,\dots, x_p \in h \Z^d} \abs{f_{i_1\dots i_p}(x, x_2, \dots, x_p)}\,.
\end{equation*}
We shall be interested in observables arising from $f \in \mathscr{B}^{(p)}$ satisfying 
\begin{equation} \label{bounded function for observable}
\limsup_{h \to 0} \, \norm{f}_1^{(h)} \;<\, \infty\,.
\end{equation}
Note that Fatou's lemma implies that $\norm{f}_1 \leq \limsup_{h \to 0} \, \norm{f}_1^{(h)}$.

We define $\mathfrak{P}_\Lambda$ as the ``polynomial'' algebra of functions on $\Xi_\Lambda$ generated by functions of the form
\begin{equation*}
M_\Lambda(f) \;\deq\; \sum_{i_1,\dots,i_p} \int_{\Lambda^p} dx_1 \cdots dx_p \; f_{i_1\dots i_p}(x_1,\dots,x_p) \, M_{i_1}(x_1) \cdots M_{i_p}(x_p)\,,
\end{equation*}
where $f \in \mathscr{B}^{(p)}$ satisfies \eqref{bounded function for observable}. $\mathfrak{P}_\Lambda$ is clearly a Poisson algebra. We equip it with the norm $\norm{A}_\infty \;\deq\; \sup_{M \in \Xi_\Lambda} \abs{A(M)}$ so that
\begin{equation} \label{norm of continuum operator}
\norm{M_\Lambda(f)}_\infty \;\leq\; \norm{f}_1\,.
\end{equation}

\subsection{Quantization}
For $f \in \mathscr{B}^{(p)}$ let us define 
\begin{equation}
\qu{S}_\Lambda(f) \;\deq\; \sum_{i_1,\dots,i_p} \sum_{x_1,\dots,x_p \in \Lambda^{(h)}} \; f_{i_1\dots i_p}(x_1,\dots,x_p) \, \qu{S}_{i_1}(x_1) \cdots \qu{S}_{i_p}(x_p)\,.
\end{equation}
If $f$ satisfies \eqref{bounded function for observable}, we find that
\begin{equation} \label{continuum norm of spin operator}
\norm{\qu{S}_\Lambda(f)} \;\leq\; \norm{f}^{(h)}_1\,.
\end{equation}
As above, quantization $\qu{\genarg}: \mathfrak{P}_\Lambda \rightarrow \qu{\mathfrak{A}}_\Lambda^{(h)}$ 
is defined by
$
\qu{M_\Lambda(f)} = \wick{\qu{S}_\Lambda(f)} 
$
and linearity. Here $\wick{\cdot}$ denotes the normal-ordering of the spin operators introduced above. Also, we set $\qu{1} = \umat$. Again, $(\qu{A})^* = \qu{\overline{A}}$.

\subsection{Dynamics in the mean-field limit} \label{sec:dyn:mf}

We consider a family $V = \pb{V^{(n)}}_{n = 1}^\infty$ of functions, where $V^{(n)} \in \mathscr{B}^{(p)}$ satisfies
\begin{equation*}
\overline{V^{(n)}_{i_1\dots i_n}(x_1, \dots, x_n)} \;=\; V^{(n)}_{\overline{i}_1 \dots \overline{i}_n}(x_1, \dots, x_n)\,,
\end{equation*}
where, we recall, $\overline{\genarg}$ on $I$ maps $(+,z,-)$ to $(-,z,+)$. We define the Hamilton function on $\Xi_\Lambda$ through
\begin{equation}\label{continuum hamilton function}
H_\Lambda \;\deq\; \sum_{n = 1}^\infty \,M_\Lambda(V^{(n)})\,.
\end{equation}
Set
\begin{equation} \label{epsilon norm of V}
\norm{V}^{(h)} \;\deq\; \sum_{n = 1}^\infty n e^n \, \norm{V^{(n)}}_{\infty, 1}^{(h)}\,.
\end{equation}
We impose the condition
$
\norm{V} \deq \limsup_{h \to 0} \,\norm{V}^{(h)} \;<\; \infty\,.
$
In the continuum limit, we observe that
\begin{equation} \label{continuum potential bound}
\sum_{n = 1}^\infty n e^n \, \sup_x \sum_{i_1, \dots, i_n} \, \int dx_2 \cdots dx_n \absb{V^{(n)}_{i_1 \dots i_n }(x, x_2, \dots, x_n)} \;\leq\; \norm{V}\,,
\end{equation}
as can be seen using Fatou's lemma. It now follows easily that, for each bounded set $\Lambda$, the sum \eqref{continuum hamilton function} converges in $\norm{\genarg}_\infty$ on $\Xi_\Lambda$ and yields a well-defined real Hamilton function $H_\Lambda$.

The Hamiltonian equation of motion reads
\begin{multline} \label{continuum Landau-Lifschitz}
\frac{d}{dt} M_i(t,x) \;=\; i \, \sum_{n = 1}^\infty \,n \sum_{i_1, \dots, i_n,j} \int dx_2 \cdots dx_n \; V^{(n)}_{i_1 \dots i_n}(x, x_2, \dots, x_n) 
\\
\tilde{\epsilon}_{i_1 i j} \, M_j(t,x) \, M_{i_2}(t,x_2) \cdots M_{i_n}(t,x_n)\,.
\end{multline}
By standard methods, we find the following global well-posedness result for \eqref{continuum Landau-Lifschitz}.

\begin{lemma} \label{continuum global well-posedness of ll}
Let $\Lambda \subset \R^d$ by any open subset of $\R^d$ and $M_0 \in \Xi_\Lambda$. Then \eqref{continuum Landau-Lifschitz} has a unique solution $M \in C^1(\R,\Xi_\Lambda)$ that satisfies $M(0) = M_0$. Moreover, the solution $M$ depends continuously on the initial condition $M_0$. Finally, we have the pointwise conservation law $\abs{M(t,x)} = \abs{M(0,x)}$ for all $t$.
\end{lemma}

\begin{remarks*}
{\em
1. As in Section \ref{section large-spin limit}, we denote the norm-preserving Hamiltonian flow by $\phi^t_\Lambda$.\\
2. Time-dependent potentials $V(t)$ may be treated exactly as in the previous section.}
\end{remarks*}

\begin{example}
{\em
Consider
\begin{equation*} 
H_\Lambda \;=\; -\int_\Lambda dx\; h(t,x) \cdot M(x)  - \frac{1}{2} \, \int_{\Lambda\times \Lambda} dx \, dy \; J(x,y) \, M(x) \cdot M(y)\,,
\end{equation*}
which yields the Landau-Lifshitz equation of motion
\begin{equation*}
\frac{d}{dt} M(t,x) \;=\; M(t,x) \wedge \qbb{h(t,x) + \int_\Lambda dy \; J(x,y) M(t,y)}\,.
\end{equation*}
}
\end{example}

\begin{remark*}
{\em
In a formal way, the Landau-Lifshitz equation (\ref{eq:LLintro2}) mentioned in the introduction can be obtained from (\ref{eq:LLintro3}) with $J=J(|x-y|)$ by Taylor expanding $M(t,y)$ up to second order in $y - x$. This leads to (\ref{eq:LLintro2}) after rescaling time by $t \mapsto \alpha t$ where $\alpha = \frac{1}{2d} \int J(|x|) \abs{x}^2\, dx$.
}
\end{remark*}

The quantum dynamics is generated by the Hamiltonian $\qu{H}_\Lambda \in \qu{\mathfrak{A}}^{(h)}_\Lambda$ defined as the quantization of $H_\Lambda$. More precisely, each term of $H_\Lambda$ is quantized and it may be easily verified that the resulting series converges in operator norm. The fact that $H_\Lambda$ is real immediately implies that $\qu{H}_\Lambda$ is self-adjoint. 
As above we introduce the short-hand notation
\begin{align*}
\alpha^t_\Lambda A &\;\deq\; A \circ \phi^t_\Lambda\,, \qquad A \in \mathfrak{A}_\Lambda,
\\
\qu{\alpha}^t_{\Lambda} \mathcal{A} &\;\deq\; U_h(t; \qu{H}_\Lambda)^* \, \mathcal{A}\, U_h(t; \qu{H}_\Lambda)\,, \qquad \mathcal{A} \in \qu{\mathfrak{A}}^{(h)}_{\Lambda} .
\end{align*}
Here, $U_h(t; \qu{H}_\Lambda)$ is the quantum mechanical propagator, equal to $e^{ish^{-d} \qu{H}_\Lambda t}$ if $\qu{H}_\Lambda$ is time-independent.

We are now in a position to state our main result on the mean-field dynamics of the quantum system on the finite lattice $\Lambda^{(h)}$ in the continuum limit, as $h \rightarrow 0$.
\begin{theorem} \label{theorem: continuum egorov for bounded domain}
Let $\Lambda \subset \R^d$ be open and bounded, $A \in \mathfrak{P}_\Lambda$ and $\epsilon > 0$. Then there exists a function $A(t) \in \mathfrak{P}_\Lambda$ such that
\begin{equation}
\sup_{t \in \R} \, \norm{\alpha_\Lambda^t A - A(t)}_\infty \;\leq\; \epsilon\,,
\end{equation}
and, for any $t \in \R$,
\begin{equation}
\normb{\qu{\alpha}^t_{\Lambda} \qu{A}  - \qu{A(t)}} \;\leq\; \epsilon + C(\epsilon, t, A) \, h^d\,,
\end{equation}
where $C(\epsilon,t,A)$ is independent of $\Lambda$.
\end{theorem}

\begin{proof}
One finds, for $f \in \mathscr{B}^{(p)}$ and $g \in \mathscr{B}^{(q)}$,
\begin{equation} \label{expanded poisson bracket}
\hb{M_\Lambda(f), M_\Lambda(g)} = pq \,M_\Lambda(f \contr g)\,
\end{equation}
where $f \contr g \in \mathscr{B}^{(p+q - 1)}$ is defined by
\begin{multline}
(f \contr g)_{i_1 \dots i_{p+q-1}}(x_1,\dots,x_{p+q-1}) 
\\
\;\deq\; i P \, \sum_{i,j} \tilde{\epsilon}_{i j i_1} f_{i i_2 \dots i_p}(x_1, \dots x_p) \, g_{j i_{p+1} \dots i_{p+q - 1}}(x_1, x_{p+1}, \dots, x_{p+q - 1})\,.
\end{multline}
We have the estimate
\begin{equation}
\norm{f \contr g}_1 \;\leq\; \norm{f}_{\infty, 1} \, \norm{g}_1\,,
\end{equation}
where
\begin{equation*}
\norm{f}_{\infty, 1} \;\deq\; \sup_{x} \sum_{i_1,\dots, i_p} \int dx_2 \dots dx_p \; \abs{f_{i_1 \dots i_p}(x, x_2, \dots x_p)}\,.
\end{equation*}

Without loss of generality, we assume that $A = M_\Lambda(f)$ for some $f \in \mathscr{B}^{(p)}$ satisfying the bound \eqref{bounded function for observable}.
Iterating
\begin{equation*}
\hb{H_\Lambda, M_\Lambda(f)} \;=\; \sum_{n = 1}^\infty np\, M_\Lambda(V^{(n)} \contr f)
\end{equation*}
we obtain that
\begin{multline*}
\hb{H_\Lambda, M_\Lambda(f)}^{(l)} \;=\; \sum_{n_1, \dots, n_l = 1}^\infty \, \qb{p n_1} \,\qb{(p+n_1 - 1) n_2} \cdots \qb{(p + n_1 + \dots + n_{l - 1} - l + 1) n_l}
\\
M_{\Lambda}\pB{V^{(n_l)} \contr \pb{V^{(n_{l - 1})} \contr \dots (V^{(n_1)}\contr f)}}\,,
\end{multline*}
with norm
\begin{align}
\normb{\hb{H_\Lambda, M_\Lambda(f)}^{(l)}}_\infty &\;\leq\; \sum_{n_1, \dots, n_l} \, \qb{p n_1} \,\qb{(p+n_1 - 1) n_2} \cdots \qb{(p + n_1 + \cdots + n_{l - 1} - l + 1) n_l}
\notag \\
&\mspace{100mu} \norm{V^{(n_l)}}_{\infty, 1} \cdots \norm{V^{(n_1)}}_{\infty, 1} \, \norm{f}_1
\notag \\
&\;\leq\; l! \sum_{n_1, \dots, n_l} \frac{(p + n_1 + \cdots + n_l)^l}{l!} \, n_1 \cdots n_l \, \norm{V^{(n_l)}}_{\infty, 1} \cdots \norm{V^{(n_1)}}_{\infty, 1} \, \norm{f}_1
\notag \\
&\;\leq\; e^p \, \norm{f}_1 \,l!\,  \qbb{\sum_n n e^n \, \norm{V^{(n)}}_{\infty, 1}}^l \,
\notag \\ \label{continuum bracket estimate}
&\;\leq\;
e^p \, \norm{f}_1 \,l!\, \norm{V}^l\,,
\end{align}
by \eqref{continuum potential bound}. Therefore, for $\abs{t} < \norm{V}^{-1}$, the series
\begin{equation} \label{classical series for continuum theory}
\sum_{l = 0}^\infty \frac{t^l}{l!}
\hb{H_\Lambda, A}^{(l)}
\end{equation}
converges in $\norm{\genarg}_\infty$ to $\alpha_\Lambda^t A$.

The quantum case is dealt with in a similar fashion, with the additional complication caused by the ordering of the generators $\{\qu{S}_i(x)\}$. This does not trouble us, however, as an exact knowledge of the ordering is not required. 
It is easy to see that, for $f$ and $g$ as above,
\begin{equation*}
ish^{-d} \qB{\qu{S}_\Lambda(f),\qu{S}_\Lambda(g)}
\end{equation*}
is equal, up to a reordering of the spin operators, to $pq \qu{S}_\Lambda(f \contr g)$. Iterating this shows that
\begin{equation*}
(ish^{-d})^l \, \qb{\qu{H}_\Lambda, \qu{A}}^{(l)}
\end{equation*}
is equal, up to a reordering of the spin operators, to
\begin{multline*}
\sum_{n_1, \dots, n_l = 1}^\infty \, \qb{p n_1} \,\qb{(p+n_1 - 1) n_2} \cdots \qb{(p + n_1 + \dots + n_{l - 1} - l + 1) n_l}
\\
\qu{S}_{\Lambda}\pB{V^{(n_l)} \contr \pb{V^{(n_{l - 1})} \contr \dots (V^{(n_1)}\contr f)}}\,,
\end{multline*}
Consequently an estimate analogous to \eqref{continuum bracket estimate} yields, for $s \geq 1$,
\begin{equation*}
\normb{(ish^{-d})^l \, \qb{\qu{H}_\Lambda, \qu{A}}^{(l)}} \;\leq\; e^p \, \norm{f}^{(h)}_1 \,l!\, (\norm{V}^{(h)})^l\,,
\end{equation*}
which readily implies the bound
\begin{equation}\label{qm series for continuum theory}
\normB{\sum_{l = 0}^\infty \frac{t^l}{l!}
\p{ish^{-d}}^l\,
\qb{\qu{H}_\Lambda, \qu{A}}^{(l)} } \leq e^p \, \norm{f}_1^{(h)}\,  \sum_{l = 0}^\infty (\abs{t} \, \norm{V}^{(h)})^l\,.
\end{equation}
If $s = 1/2$, the first line of \eqref{continuum bracket estimate} gets the additional factor $\sqrt{2}^{n_1 + \cdots + n_l + p}$. This may be dealt with by replacing the factor $(p + n_1 + \cdots + n_l)^l$ in the second line of \eqref{continuum bracket estimate} with $(rp + rn_1 + \cdots + rn_l)^l/r^l$. The desired bound then follows for $0 < r \leq 1 - \frac{1}{2} \log 2$. Note that in this case the convergence radius for $t$ is reduced to $r \norm{V}^{-1}$. For ease of notation, we restrict the following analysis to the case $s \geq 1$, while bearing in mind that the extension to $s = 1/2$ follows by using the above rescaling trick.

Now, by definition of $\norm{V}$, for any $\abs{t} < \norm{V}^{-1}$ there is an $h_0$ such that \eqref{qm series for continuum theory} converges in norm to $\qu{\alpha}^t_{\Lambda} \qu{A}$ for all $h \leq h_0$, uniformly in $h$ and $\Lambda$.

In order to establish the statement of the theorem for short times $\abs{t} < \norm{V}^{-1}$, we remark that the commutation relations \eqref{continuum commutation relations in ladder rep} imply the bound
\begin{equation*}
\normb{\mathcal{A} - \mathcal{B}} \;\leq\; \frac{h^d}{s} \, p^2 \, \norm{f}^{(h)}_1\,,
\end{equation*}
for arbitrary reorderings, $\mathcal{A}$ and $\mathcal{B}$, of the same operator $\qu{S}_\Lambda(f)$, with $f \in \mathscr{B}^{(p)}$ for some $p < \infty$.

If we define $\qu{\alpha^t_\Lambda A}$ through its norm-convergent power series, we therefore get
\begin{align*}
&\normb{\qu{\alpha}^t_{\Lambda} \qu{A} - \qu{\alpha^t_\Lambda A}} 
\\
&\leq\;\frac{h^d}{s} \, \sum_{l = 0}^\infty \frac{\abs{t}^l}{l!}
 \sum_{n_1, \dots, n_l} \, \qb{p n_1} \,\qb{(p+n_1 - 1) n_2} \cdots \qb{(p + n_1 + \cdots + n_{l - 1} - l + 1) n_l}
\\
&\qquad 
(n_1 + \cdots + n_l -l+1)^2\,\norm{V^{(n_l)}}^{(h)}_{\infty, 1} \cdots \norm{V^{(n_1)}}^{(h)}_{\infty, 1} \, \norm{f}^{(h)}_1
\\
&\leq\;\frac{h^d}{s} \, \sum_{l = 0}^\infty \abs{t}^l
 \sum_{n_1, \dots, n_l} \frac{(p + n_1 + \cdots + n_l)^{l+2}}{l!} \, n_1 \cdots n_l \, \norm{V^{(n_l)}}^{(h)}_{\infty, 1} \cdots \norm{V^{(n_1)}}^{(h)}_{\infty, 1} \, \norm{f}^{(h)}_1
\\
&\leq\;\frac{h^d}{s} \, \sum_{l = 0}^\infty \abs{t}^l
e^p \, \norm{f}^{(h)}_1 \,(l+2) (l+1)\,  \qbb{\sum_n n e^n \, \norm{V^{(n_1)}}^{(h)}_{\infty, 1}}^l
\\
&\leq\;\frac{h^d}{s} \, e^p \, \norm{f}^{(h)}_1 \, \sum_{l = 0}^\infty (l+2) (l+1) \, (\abs{t} \, \norm{V}^{(h)})^l
\\
&=\; O(h^d)\,,
\end{align*}
where in the last step we have used the fact that the sum convergences uniformly in $h$, for $h$ small enough, as seen above.

Arbitrary times are reached by iteration of the above result.
\end{proof}

\subsection{The thermodynamic limit}
The above result may again be formulated in the thermodynamic limit as $\Lambda \rightarrow h\Z^d$. We only sketch the arguments, which are almost identical to those of Section \ref{section: thermodynamic limit for lattice}.

The quantum quasi-local algebra is
\begin{equation*}
\qu{\mathfrak{A}}^{(h)} \;\deq\; \overline{\bigvee_{\Lambda \subset \subset \R^d} \qu{\mathfrak{A}}^{(h)}_\Lambda}\,,
\end{equation*}
The existence of dynamics is guaranteed by the following statement.
\begin{lemma} \label{continuum quantum limit dynamics}
Let $h > 0$ and suppose $\mathcal{A} \in \qu{\mathfrak{A}}^{(h)}_{\Lambda_0}$ for some bounded and open $\Lambda_0 \subset  \R^d$. Then, for any $t \in \R$, the following limit exists in the norm sense:
\begin{equation*}
\lim_{\Lambda \to \infty} \qu{\alpha}_\Lambda^t \mathcal{A} \;\eqd\; \qu{\alpha}^t \mathcal{A} \,,
\end{equation*}
By continuity this extends to a strongly continuous one-parameter group $(\qu{\alpha}^t)_{t \in \R}$ of automorphisms of $\qu{\mathfrak{A}}^{(h)}$.
\end{lemma}

The classical quasi-local algebra is
\begin{equation*}
\mathfrak{A} \;\deq\; \overline{\bigvee_{\Lambda \subset\subset  \R^d} \mathfrak{P}_\Lambda}\,.
\end{equation*}

\begin{lemma}
Let $A \in \mathfrak{P}_{\Lambda_0}$ for some open and bounded $\Lambda_0 \subset \subset \R^d$. Then, for any $t \in \R$, the following limit exists in $\norm{\genarg}_\infty$:
\begin{equation*}
\lim_{\Lambda \to \infty} \alpha_\Lambda^t A \;\eqd\; \alpha^t A \,,
\end{equation*}
By continuity this extends to a strongly continuous one-parameter group $(\alpha^t)_{t \in \R}$ of automorphisms of $\mathfrak{A}$. Furthermore, $\alpha^t A = A \circ \phi^t$, where $\phi^t = \phi^t_{\R^d}$ is the Landau-Lifschitz flow defined in Lemma \ref{continuum global well-posedness of ll}.
\end{lemma}

Now, for $f \in \mathscr{B}^{(p)}$, $M(f)$ and $\qu{S}(f)$ are well-defined in the obvious way.
Define $\mathfrak{P}$ as the algebra generated by functions of the form $M(f)$, where $f$ satisfies \eqref{bounded function for observable}.

\begin{theorem}
Let $A \in \mathfrak{P}$ and $\epsilon > 0$. Then there exists a function $A(t) \in \mathfrak{P}$ such that
\begin{equation}
\sup_{t \in \R} \, \norm{\alpha^t A - A(t)}_\infty \;\leq\; \epsilon\,,
\end{equation}
and, for any $t \in \R$,
\begin{equation}
\normb{\qu{\alpha}^t\qu{A}  - \qu{A(t)}} \;\leq\; \epsilon + C(\epsilon,t,A) \, h^d\,.
\end{equation}
\end{theorem}

\subsection{Evolution of coherent states}
In this section, our ``smearing functions'' $f$ are assumed to have compact support, i.\,e.\  to belong to the space
\begin{equation*}
\mathscr{B}^{(p)}_c \;\deq\; \mathscr{B}^{(p)} \cap C_c(\R^{pd};\C^{3^p})\,.
\end{equation*}
In addition, we require the interaction potential $V$ to be of finite range in the sense that there exists a sequence $R_n > 0$ such that if $\abs{x_i - x_j} > R_n$ for some pair $(i,j)$ then $V^{(n)}_{i_1 \dots i_n}(x_1, \dots, x_n) = 0$.

Next, we take some initial classical spin configuration $M \in C(\R^d; \mathbb{S}^2)$, or, more generally, a function $M : \R^d \to \mathbb{S}^2$ whose points of discontinuity form a null set. We shall study the time evolution of product states $\rho_M$ on $\qu{\mathfrak{A}}^{(h)}$ that reproduce the given classical state $M$. For open and bounded $\Lambda \subset \R^d$, we define the product state
\begin{equation*}
\ket{M_\Lambda} \deq \bigotimes_{x \in \Lambda^{(h)}} \ket{M(x)}\,,
\end{equation*}
where $\ket{M(x)}$ is the coherent spin state corresponding to the unit vector $M(x)$. For $\mathcal{A} \in \qu{\mathfrak{A}}_\Lambda^{(h)}$, define
\begin{equation*}
\rho_M\pb{\mathcal{A}} \;\deq\; \scalar{M_\Lambda}{\mathcal{A} \, M_\Lambda}\,,
\end{equation*}
which we extend to arbitrary $\mathcal{A} \in \qu{\mathfrak{A}}^{(h)}$ by continuity.

For our main result on the time evolution of coherent states, we first record the following auxiliary result whose elementary proof we omit.
\begin{lemma} \label{lemma: continuum convergence}
Let $f \in \mathscr{B}^{(p)}_c$ satisfy \eqref{bounded function for observable}. Then
\begin{equation}
\lim_{h \to 0} \rho_M \pb{\qu{S}(f)} \;=\; M(f)\,.
\end{equation}
\end{lemma}

The last result in this paper links the quantum time evolution of coherent spin states with the classical evolution in the mean-field/continuum limit when the lattice spacing $h$ tends to 0.
\begin{theorem}
Let $t \in \R$, $A \in \mathfrak{P}$ and $M$ be as described above. Let $M(t)$ be the solution of \eqref{continuum Landau-Lifschitz} on $\R^d$ with initial configuration $M$. Then
\begin{equation*}
\lim_{h \to 0} \rho_M \pb{\qu{\alpha}^t \qu{A}} \;=\; A(M(t))\,,
\end{equation*}
uniformly in $t$ on compact time intervals.
\end{theorem}
\begin{proof}
The proof is a corollary of the proof of Theorem \ref{theorem: continuum egorov for bounded domain}. First, let $\abs{t} < \norm{V}^{-1}$ and pick an $\epsilon > 0$. Choose a cutoff such that the tails of the thermodynamic limits of the series \eqref{classical series for continuum theory} and \eqref{qm series for continuum theory} are bounded by $\epsilon$. We therefore have to estimate a finite sum of terms of the form
\begin{equation*}
\absb{\rho_M \pb{\qu{S}(g)} - M(g)}\,,
\end{equation*}
where $g \in \mathscr{B}^{(p(g))}_c$ because of our assumptions on $V$.
By Lemma \ref{lemma: continuum convergence}, for $h$ small enough, these are all bounded by $\epsilon$, and the claim for small times follows. Finally, by iteration, we extend the result to arbitrary times.  \end{proof}

\noindent
{\bf Acknowledgement.} We thank Simon Schwarz for very useful discussions.

\end{document}